\documentclass[12pt,journal,draftcls,letterpaper,onecolumn]{IEEEtran}

\usepackage{graphicx}  
\usepackage{amsmath}
\usepackage{amssymb}

\usepackage{algorithmic}
\usepackage{xcolor}
\usepackage{color}
\newtheorem{theorem}{Theorem}
\newtheorem{lemma}{Lemma}

\newtheorem{proposition}{Proposition}

\definecolor{darkgreen}{rgb}{0.1,0.5,.5}
\definecolor{darkred}{rgb}{0.8,0,0}
\definecolor{teal}{rgb}{0.05,0.32,0.41}
\definecolor{darkblue}{rgb}{0,0.0,0.5}
\definecolor{blackgreen}{rgb}{0,0.4,0}
\definecolor{purple}{rgb}{0.5,0,0.3}
\definecolor{grey}{rgb}{0.7,0.5,0.5}
\definecolor{orange}{rgb}{0.6,0.4,0.1}

\begin{document}
%
\title{Uplink Time Scheduling with Power Level Modulation in Wireless Powered Communication Networks}
%
%
\author{Yeongwoo Ko, Sang-Hyo Kim, and Jong-Seon No
\thanks{Y. Ko and J.-S. No are with the Department of Electrical and Computer
Engineering, INMC, Seoul National University, Seoul 08826, Korea, e-mail:
kyw1623@ccl.snu.ac.kr, jsno@snu.ac.kr. S.-H. Kim (corresponding author) is with College of Information and Communication Engineering, Sungkyunkwan University, Suwon 16419, Korea, email: iamshkim@skku.edu.}
}

%
%
%
\markboth{Submission for...}{Regular Paper}



\maketitle

\begin{abstract}

In this paper, we propose downlink signal design and optimal uplink scheduling for the wireless powered communication networks (WPCNs). Prior works give attention to resource allocation in a static channel because users are equipped with only energy receiver and users cannot update varying uplink schedulling. For uplink scheduling, we propose a downlink signal design scheme, called a power level modulation, which conveys uplink scheduling information to users. First, we design a downlink energy signal using power level modulation. Hybrid-access point (H-AP) allocates different power level in each subslot of the downlink energy signal according to channel condition and users optimize their uplink time subslots for signal transmission based on the power levels of their received signals. Further, we formulate the sum throughput maximization problem for the proposed scheme by determining the uplink and downlink time allocation using convex optimization problem. Numerical results confirm that the throughput of the proposed scheme outperforms that of the conventional schemes.

\end{abstract}

\begin{keywords}

Energy harvesting, hybrid-access point (H-AP), resource management, wireless powered communication network (WPCN).

\end{keywords}

\vspace{10pt}
\section{Introduction}

Energy durability has always been an important research topic for wireless communicaiton networks. Many researches have proposed energy saving protocols and transmission schemes for longer battery life. Recently, wireless powered communication networks (WPCNs) have been proposed for wireless communication environment such as sensor networks or internet of things (IoT) \cite{rf1}, \cite{rf2}. Such radio-frequency energy harvesting (RF-EH) has emerged as an alternative solution for prolonging the lifetime of wireless devices.
 
There are two main research categories for RF-EH based wireless communication systems, which are simultaneous wireless information and power transfer (SWIPT) and WPCN. In the SWIPT, wireless energy transfer (WET) and wireless information transmission (WIT) are simultaneously accomplished in the downlink (DL) \cite{swipt1}-\cite{swipt3}, but, there is no downlink data transmission in WPCN. Hybrid-access point (H-AP) broadcasts energy transfer RF signals in the downlink phase and users charge their energy storages such as batteries or supercapacitors from received WET signals. In the uplink (UL) phase, users transmit their WIT signals to the H-AP using harvested energy. 

 A harvest-then-transmit protocol was proposed for an efficient establishment of WPCN in wireless sensor networks \cite{wpcn1}, where the H-AP broadcasts the WET signal in the downlink phase and recovers the message from the WIT signal in the uplink phase. Time division multiple access (TDMA) was adopted in the uplink, where time slots are optimally allocated to each user for maximizing the uplink throughput.
 In \cite{wpcn_1}, a general scheme for wireless powered cellular networks was studied, which incorporates a two-way information transmission together with energy transfer in the downlink from the H-AP to the cellular users.
 Heterogeneous WPCNs were also studied, where EH nodes and non-EH nodes coexist in \cite{wpcn2} and WPCN with full-duplex was considered in \cite{wpcn3}. Most of conventional works focus on optimization with single block time in fixed channel state and the uplink scheduling is not considered in quasi-static channel \cite{wpcn1}-\cite{wpcn3}. For WPCN with energy storage constraint, the optimal downlink power allocation policy was proposed in \cite{wpcn5}, where the transmitted downlink signals have different power levels in the time slots. The H-AP concentrates all available energy in the first few downlink time slots but it is not intended to convey uplink scheduling information to the users. However, in ordinary wireless powered sensor networks, users are not assumed to have a downlink data channnel receiver and thus uplink time slots cannot be dynamically scheduled, which results in inefficient utilization of resources.
 
 In this paper, we focus on how users can dynamically provide uplink scheduling according to quasi-static channel. We propose the power level modulation scheme for donwlink WET signal at H-AP so as to provide uplink scheduling information to the users, where users extract their allocated uplink time slots from the received power level-modulated WET signal. We formulate the sum-rate optimization problem in this situation by the convexity of the problem. The reduced dynamic range is introduced to alleviate the performance loss due to the peak power constraint. Numerical analysis shows that the throughput of the proposed WPCN with power level modulation scheme outperforms that of the conventional schemes in the quasi-static channel condition.

 This paper is organized as follows. In Section II, we describe the system model. First, the proposed downlink signal design is introduced and the optimization problem formulation and its convexity proof follow in Section III. Numerical analysis for the proposed WPCN with power level modulation is given in Section IV. Finally, conclusions are given in Section V.

\vspace{10pt}
\section{System Model} \label{System Model}

\begin{figure}
\centering
\includegraphics[scale=0.6]{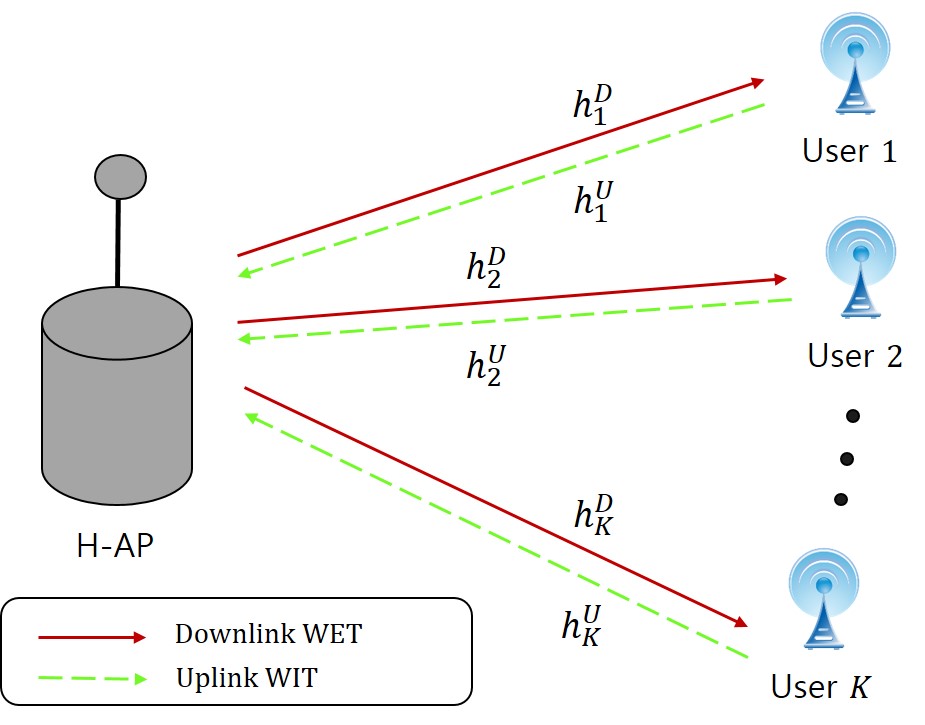}
\caption{System model: Wireless powered communication networks.}
\label{grp1}
\end{figure}

 WPCN with one H-AP and $K$ users is shown in Fig. 1, where it operates in a time division duplexing (TDD) manner with WET in the DL and WIT in the UL. H-AP broadcasts a WET signal to the users and then each user transmits its WIT signal to the H-AP by utilizing the harvested energy. All the energy and information transfers are operated via TDD over the same frequency band to attain high spectrum efficiency. 

 It is assumed that the total operation time for WPCN is normalized as one. Quasi-static fading channel is considered, that is, the DL and  the UL channel states are constant over a communication frame and it is assumed that each channel state is known at the H-AP. 
In the conventional WPCN, the DL signal transmits a WET signal with a fixed average power intensity during the allotted time slot and there is no transmission of uplink scheduling information. For synchronous transmission, the user time slots should be uniform and predetermined in the UL. In this paper, we propose to divide the downlink time slot into a number of subslots and to power level-modulate the WET signal and the scheduling information for each user in the UL is transmitted to the users by the power level-modulated WET signal.

 The DL and UL signal frame structure of WPCN is depicted in Fig. 2. In the conventional WPCN with assumption that the users cannot be dynamically scheduled in the UL for low cost, the DL transmission power is constant over the DL time slot $\tau^D$ and the UL transmission signals of $K$ users are transmitted over the equal time subslots ${{1-\tau^D} \over {K}}$ with power $P_{i}^U$. In the proposed WPCN, the DL energy transfer signals are transmitted over $K$ subslots with the different powers $P_{i}^D, ~i=1,2,...,K$ and the time subslot ${\tau^D}\over{K}$. The UL transmission time slot is divided into $K$ time subslots, $\tau_{1}^U, \tau_{2}^U, ... , \tau_{K}^U$ and each user transmits its information signal to H-AP with transmit power $P_{i}^U$ in the subslot $\tau_i^U, ~i=1,2...,K$.

 It is assumed that the H-AP has an energy constraint $E^D$ in a single time slot of the DL. The peak power constraint and the average power are denoted by $P_P$ and $P_A$, respectively, so that the downlink time slot is bounded as
\begin{equation}
\tau^D \leq {{E^D} \over {P_A}}.
\end{equation}

 The DL channel power gain from the H-AP to user $i$ is denoted by $h_{i}^D$. It is assumed that energy harvesting due to receiver noise is negligible compared to the sufficiently large $P_{i}^D$. Then, the amount of harvested energy $E_{i,j}$ of user $i$ in the $j$-th DL time subslot and the total harvested energy $E_i$ at user $i$ are given as 
\begin{equation}
\begin{aligned}
E_{i,j}&={\eta}h_{i}^DP_{i}^D\tau^{D}\frac{1}{K}\\
E_{i}&=\sum_{j=1}^{K}{\eta}h_{i}^DP_{j}^D\tau^{D}\frac{1}{K} \leq {\eta}h_{i}^DP_{A}\tau^{D}, ~~ i=1,2,...,K,
\end{aligned}
\end{equation}
where $\eta$ is the energy harvesting efficiency at all users. In the UL WIT phase, users transmit their information signals to the H-AP in TDMA manner. The total consumed power at user $i$ and the channel power gain from user $i$ to the H-AP are $P_{i}^U+p^c_i$ and $h_{i}^U$, respectively, where $p^c_i$ is the circuit power dissipation of user $i$. The achievable rate of user $i$ measured in nats/s/Hz is expressed as 
\begin{equation}
\begin{aligned}
R_i={\tau_{i}^U}\log_{}{(1+\frac{h_{i}^UP_{i}^U}{N_0})}, ~~ i=1,2,...,K,  \\
\end{aligned}
\end{equation}
where  $N_0$ is the one-sided power spectral density of the additive white Gaussian noise.

\vspace{10pt}

\begin{figure}
\centering
\includegraphics[scale=0.6]{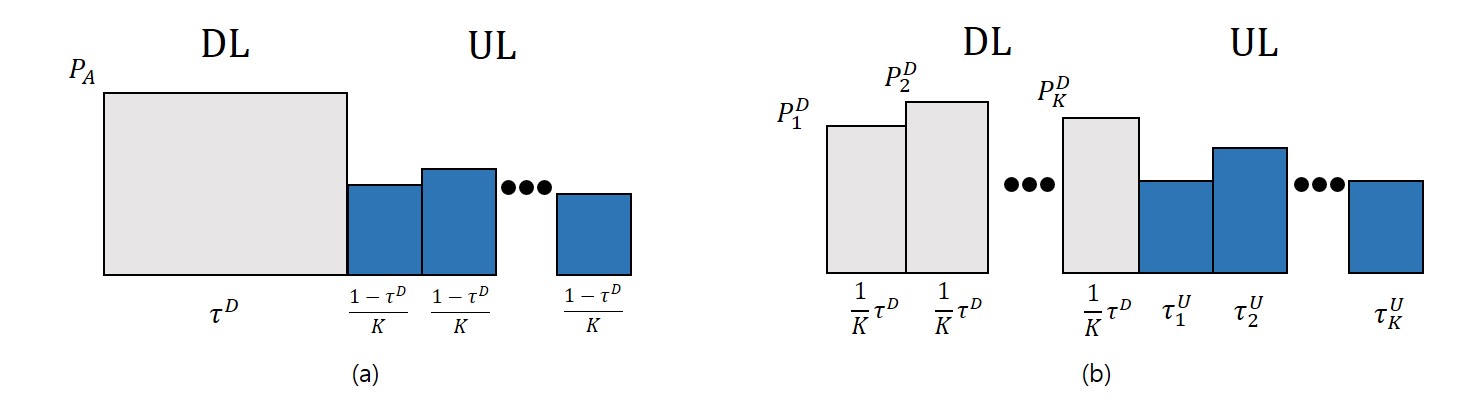}
\caption{DL and UL signal frames; (a) Conventional WPCN, (b) Proposed WPCN.}
\label{grp1}
\end{figure}

\section{Power Level Modulation Scheme and Problem Formulations with Convexity}
In this section, we propose power level modulation schemes and an optimization problem of the schemes. We optimize downlink time, uplink times, and power levels in order to maximize the total sum rate of the WPCN and prove a convexity of the problem.
\subsection{Sum Rate Maximization for Power Level Modulation Scheme with Peak Power Constraint}
 In the proposed WPCN scheme, DL time slot is uniformly divided into $K$ subslots. Each DL subslot is then equal to $\frac{\tau^D}{K}$ and H-AP transmits energy signals using power $P_i^D, i=1,2,...,K$ in the $i$-th DL subslot. Since the channels from H-AP to users are not homogeneous, the amount of harvested energies of the users are different from each other. However, the relative ratio of power levels $P_{i}^D,~i=1,2,...,K$ of the received WET signals from the H-AP is preserved for all users. Furthermore, it is assumed that users can calculate their harvested energy for each time subslot accurately. 
Thus, even though the channel condition changes, the time duration $\tau_{i}^U$ allocated to user $i$ in the UL time slot is computed using the relative ratio of $E_{i,j}$, that is,
\begin{equation}
\begin{aligned}
&\tau_{1}^U:\tau_{2}^U:...:\tau_{K}^U=E_{i,1}:E_{i,2}:...:E_{i,K},\\
&\sum_{j=1}^{K}{\tau_j^U}=1-\tau^D.
\end{aligned}
\end{equation}
Then, we have
\begin{equation}
\begin{aligned}
&\tau_{i}^U=\frac{E_{i,i}}{\sum_j^{K}{E_{i,j}}}{(1-\tau^D)}.
\end{aligned}
\end{equation}
The optimal values of $\tau^D$ and $P_i^D$ are obtained by solving the following optimization problem at H-AP.
Let $R_{sum}$ be the sum rate defined as
\begin{equation}
R_{sum}=\sum_{i=1}^{K}R_i
\end{equation}
and let $\pmb{P}=[P_1^D,P_{2}^D, ... ,P_{K}^D, P_{1}^U,P_{2}^U, ... ,P_{K}^U]$ and $\pmb{\tau}=[\tau^D,\tau_{1}^U,\tau_{2}^U, ... , \tau_{K}^U]$. Then, we can formulate the optimization problem as follows:

\begin{equation}
\begin{aligned}
\max_{\pmb {P,\tau}}{~R_{sum}}&\\
\text{such}~\text{that}~ &C1{:}~\tau^D \geq 0, \;\;\; \tau_{i}^U \geq 0, \;\;\; P_{i}^D \geq 0,\\
&C2{:}~\tau^D+\tau_{1}^U+...+\tau_{K}^U \leq 1,\\
&C3{:}~P_{1}^D+P_{2}^D+...+P_{K}^D\leq KP_A,\\
&C4{:}~P_{i}^D\leq P_P,\\
&C5{:}~{(P_{i}^U+p^c_i)}\tau_{i}^U\leq \sum_{j=1}^{K}{\eta h_{i}^DP_{j}^D{{\tau^D}\over{K}}},\\
&C6{:}~\tau_{i}^U\leq {{P_{i}^D}\over {\sum_{j=1}^{K}{P_{j}^D}}}{(1-\tau^D)}, \;\;\;\;\;\;i=1,2,...,K.
\end{aligned}
\end{equation}

The constraint $C5$ implies that the total energy consumed by the $i$-th user cannot exceed the harvested energy in the UL phase. The condition for UL time allocation of the $i$-th user through power level modulation of the DL signal is given by $C6$.

\begin{lemma}The objective function $R_{sum}$ is a concave function of $\tau^D$.
\end{lemma}
\begin{proof}Let $R'=\log(1+\mu\tau^D\sum_{j=1}^{K}{P_{j}^D})$ . Since the logarithm of an affine function is concave, $R'$ is a concave function of $\tau^D$. Then, $R_{sum}$ is a perspective function of $R'$ and it preserves the concativity of original function $R'$ \cite{lse}. \end{proof}

The above optimization problem has multi variables $P_{i}^D$, $\tau^{D}$, and $\tau_{i}^U$. Since these variables are coupled in the constraints $C5$ and $C6$, it is not straightforward to find the globally optimal solution to this problem. A proper transformation of the problem can lead to the straightforward proof of the convexity of the problem. The above optimization problem can be transformed using the variable substitution proposed in \cite{tr1}. First, note that the variables $P_{i}^D$, $\tau^D$, and $\tau_{i}^U$ are positive. Then, we can take variable substitution as $P_{i}^D\triangleq \exp(p_{i}^D)$, $P_{i}^U\triangleq \exp(p_{i}^U)$, $\tau^{D}\triangleq \exp(t^{D})$, and $\tau_{i}^U\triangleq \exp(t_{i}^U)$ for all $i=1,2,...,K$. Let $\pmb{p}=[p_{1}^D,p_{2}^D, ... , p_{K}^D, p_{1}^U,p_{2}^U, ... ,p_{K}^U]$, and $\pmb{t}=[t^D,t_{1}^U,t_{2}^U, ... , t_{K}^U]$. Then, (7) can be rewritten as

\begin{equation}
\begin{aligned}
\max_{\pmb{p,t}}{~R_{sum}}&\\
\text{such~that~} &C1{:}~\exp(t^D)+\exp(t_{1}^U)+\exp(t_{2}^U)+...+\exp(t_{K}^U) \leq 1, \\
&C2{:}~\exp(p_{1}^D)+\exp(p_{2}^D)+...+\exp(p_{K}^D) \leq KP_A, \\
&C3{:}~\exp(p_{i}^D) \leq P_P,\\
&C4{:}~(\exp(p_{i}^U)+p^c_i)\exp(t_{i}^U)\leq \sum_{j=1}^{K}{{{\eta h_{i}^D}\over{K}}\exp(p_{j}^D+t^D)},\\
&C5{:}~\exp(t_{i}^U)\leq {{\exp(p_{i}^D)}\over {\sum_{j=1}^{K}{\exp(p_{j}^D)}}}{(1-\exp(t^D))}, \;\;\;\;\;\;i=1,2,...,K.
\end{aligned}
\end{equation}

\begin{theorem} The optimization problem in (7) is convex.
\end{theorem}

\begin{proof} The objective function $R_{sum}$ is concave and the constraints C1, C2, and C3 are convex. For convenience of proof, we take the logarithm on both sides of C4 and then we have
\begin{equation}
\log(\exp(p_{i}^U+t_{i}^U)+p^c_i\exp(t_{i}^U))\leq\sum_{j=1}^{K}{(p_{j}^D+t^D)\log({{\eta h_{i}^D}\over{K}}}).
\end{equation}
 The right-hand side is a linear sum of the variables $P_j^D,t^D$ multiplied by a constant $\log({{\eta h_{i}^D}\over{K}})$ and the left-hand-side is a log-sum-exp function that has been proved to be convex in \cite{lse}. Finally the inequality in C5 can be rewritten as 
\begin{equation}
\exp(t_{i}^U+p_{1}^D)+\exp(t_{i}^U+p_{2}^D)+...+\exp(t_{i}^U+p_{K}^D)+\exp(t^D+p_{i}^D)\leq \exp(p_{i}^D).
\end{equation}
The right-hand side is convex and the left-hand side is a linear sum of convex function and thus it is also convex.
\end{proof}
The convexity of sum rate maximization is proved and thus the maximum value of sum rate can be numerically obtained.

\subsection{Modified Sum Rate Maximization for Power Level Modulation Scheme with Reduced Dynamic Range}
 In general, maximal sum rate with peak power constraint C3 is lower than that with no peak power constraint in the power level modulation for uplink scheduling. Now, we introduce reduced dynamic range for the power level modulation in H-AP as 
\begin{equation}
\begin{aligned}
&P_A(1-\alpha) \leq P_i^* \leq P_P,  ~~ i=1,2,...K, ~~ 0\leq \alpha \leq 1 \\
&P_i^*=P_i^D-P_A(1-\alpha),
\end{aligned}
\end{equation}
where $\alpha$ denotes the dynamic range index. 

Fig. 3 shows the power level modulation with reduced dynamic range, where $\alpha$ is known to all users. Thus, $\alpha=1$ means that there is no limitation of dynamic range as in (7) and  $\alpha=0$ means the conventional WPCN with constant power level.
\begin{figure}
\centering
\includegraphics[scale=0.5]{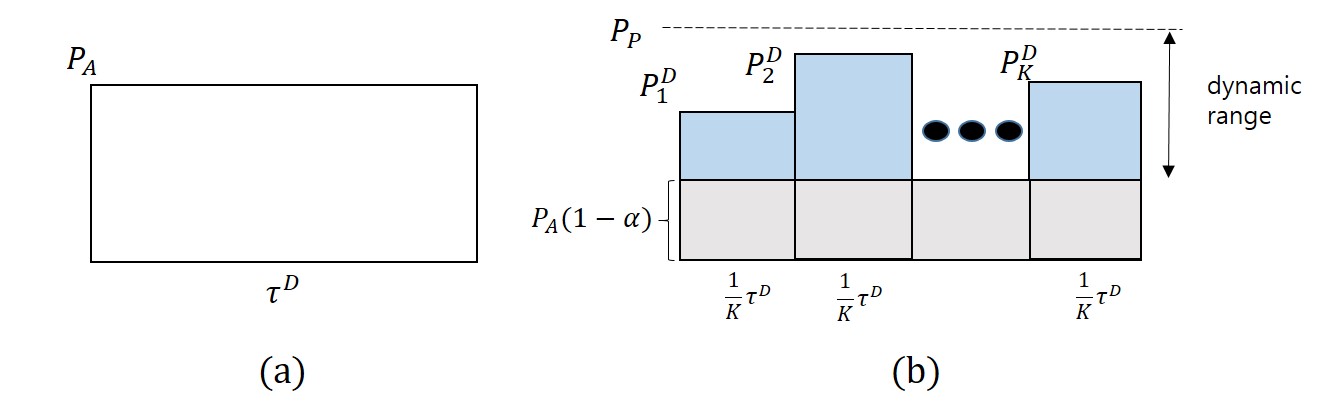}
\caption{Downlink WET signals; (a) Conventional WPCN, (b) Proposed WPCN with reduced dynamic range.}
\label{fig_systemmodel1}
\end{figure}

The following two propositions give upper bounds on $\alpha$ and $K$.

\begin{proposition}The condition for dyamic range index $\alpha$ such that it is not affected by C4 is given as
\begin{equation}
\begin{aligned}
\alpha \leq {{{{P_P}\over{P_A}}-1}\over{K-1}}.
\end{aligned}
\end{equation}
\end{proposition}

\begin{proof} Consider the case that the H-AP allocates all UL time to a single user $i$ because its channel is extremely  better than other users' channels. The transmission power $P_i^D$ takes its maximum possible value in the DL time slot corresponding to the user and the transmission power $P^D_i$ of the user $i$ in the $i$-th time slot becomes $\alpha P_A K+(1-\alpha)P_A$. Then, the condition for not being affected by constraint C5 is $P^D_i \leq P_P.$ Thus, (12) can be derived. 
\end{proof}

\begin{proposition} For the case that dynamic range index, peak power, and average power constraints are given, the number of users that is not affected by C4 is also given as
\begin{equation}
\begin{aligned}
K  \leq \lfloor {{{{P_P}\over{P_A}}-1}\over{\alpha}}+1\rfloor.
\end{aligned}
\end{equation}
\end{proposition}
\begin{proof} Proof can be done similarly to $Proposition~1$.
\end{proof}

Similar to the case in Section 3. A, each user can determine the uplink time $\tau_i^U$ based on the ratio of $E_{i,j}$ and $\alpha$, that is,
\begin{equation}
\begin{aligned}
&\tau_{1}^U:\tau_{2}^U:...:\tau_{K}^U=(E_{i,1}-E_{Ci}):(E_{i,2}-E_{Ci}):...:(E_{i,K}-E_{Ci}),\\
&\sum_{j=1}^{K}{\tau_j^U}=1-\tau^D,\\
\text{where}~~~~~~~~~ &E_{Ci}={{1-\alpha}\over{K}}\sum_{j=1}^{K}{E_{i,j}}.
\end{aligned}
\end{equation}
Then, the uplink time for user $i$ is then calculated as 
\begin{equation}
\begin{aligned}
&\tau_{i}^U=\frac{E_i-{{\sum_i^{K}{E_i}(1-\alpha)}\over{K}}}{\alpha\sum_i^{K}{E_i}}{(1-\tau^D)}.
\end{aligned}
\end{equation}

Similar to (7), the problem that H-AP should solve for optimal uplink scheduling of modified sum rate maximization for the power level modulation with reduced dynamic range is transformed into:

\begin{equation}
\begin{aligned}
\max_{\pmb {P,\tau}}{~R_{sum}}&\\
\text{such~that~} &C1{:}~ \tau^D \geq 0, \;\;\; \tau_{i}^U \geq 0, \;\;\; P_{i}^D \geq 0,\\
&C2{:}~\tau^D+\tau_{1}^U+...+\tau_{K}^U \leq 1,\\
&C3{:}~ P_{1}^D+P_{2}^D+...+P_{K}^D\leq KP_A,\\
&C4{:}~ P_{i}^D\leq P_P,\\
&C5{:}~ (P_{i}^U+p^c_i)\tau_{i}^U\leq \sum_{j=1}^{K}{\eta h_{i}^DP_{j}^D{{\tau^D}\over{K}}},\\
&C6{:}~ \tau_{i}^U\leq {{P_{i}^D-(1-\alpha)P_A}\over \alpha {\sum_{j=1}^{K}{P_{j}^D}}}{(1-\tau^D),}\;\;\;\;\;\;i=1,2,...,K.
\end{aligned}
\end{equation}

$Theorem~2:$ The problem in (16) is convex.

$Proof :$ The convexity is proved similarly to $Theorem~1$.~~~~~~~~~~~~~~~~~~~~~~~~~~~~~~~~~~~~~~~~~~~~~~~~~~~~~$\blacksquare$ 
Since the problem in (16) is convex, the maximum value of sum rate with limited dynamic range can also numerically be obtained.

\vspace{10pt}
\section{Simulation Results}

\begin{figure}
\centering
\includegraphics[scale=1.2]{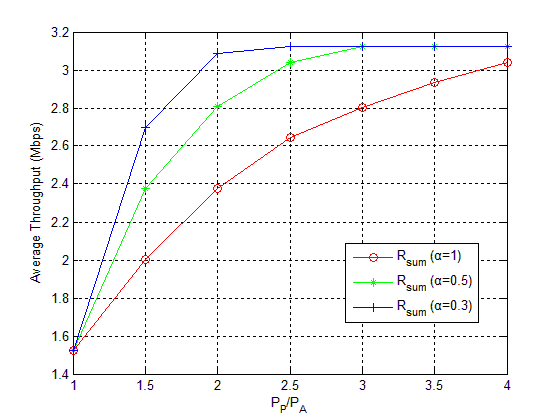}
\caption{Comparison of average sum rate of WPCN with $P_A=20$ and $K=5$.}
\label{grp1}
\end{figure}

\begin{figure}
\centering
\includegraphics[scale=1.2]{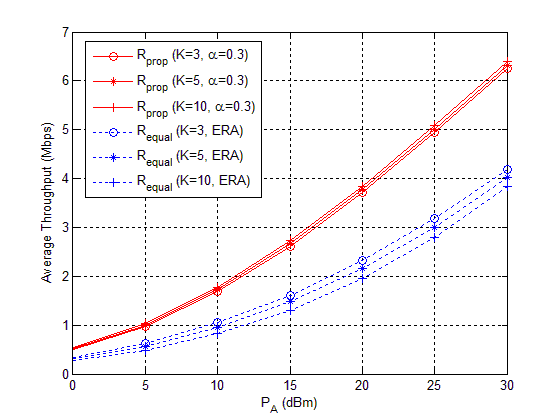}
\caption{Comparison of average sum rate of WPCN with $P_P/P_A=4, \alpha=0.3$, and distances $\{5, 10, 15\}$ for $K=3$ and $\{5,10,15,10,10,...\}$ for $K=5,10$.}
\label{grp1}
\end{figure}

\begin{figure}
\centering
\includegraphics[scale=1.2]{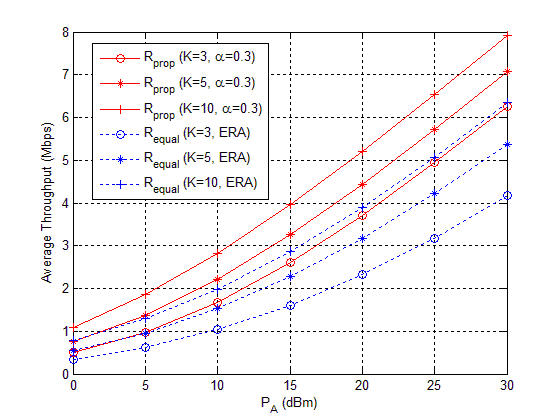}
\caption{Comparison of average sum rate of WPCN with $P_P/P_A=4, \alpha=0.3$, and distances $\{5, 10, 15\}$ for $K=3$ and $\{5,10,15,10,10,...\}$ for $K=5,10$.}
\label{grp1}
\end{figure}

 In this section, we provide numerical results to show the performance of the proposed WPCN schemes. In the network setting for simulation, we adopt the distance-dependent path loss model and assume a quasi-static Rayleigh fading channel \cite{path}. Therefore, the DL channel power gain for user $i$ is given as $h_i^D=\beta^2_iD^{-\gamma}_i,$ where $\beta^2_i$ is an exponential random variable with unit mean, $D_i$ is the distance of H-AP from user $i$, and the path loss exponent $\gamma$ is set to two.
 The channel variables are assumed to be i.i.d. at all nodes. Since the forward and reverse links are reciprocal, $h_{i}^U=h_{i}^D$ within a block time. Harvesting efficiency $\eta$ is set to 0.5 as in \cite{effi}. The noise power $N_0$ is assumed to be $-160$ [dBm]. Power consumed by circuit is set as $p^c_i=0$[W]. The simulation results are averaged over 1000 channel realization.

 In Fig. 4, we consider a five-user WPCN and set  $P_A=20$[dBm/Hz], $K=5$, and $\alpha=0.3, 0.5, 1$, where the average throughput versus the ratio of peak power and average power $P_P/P_A$ is plotted. It is shown that the proposed WPCN with lower dynamic range index, that is, small value of $\alpha$  achieves higher throughput.

 The average throughput for the proposed uplink time scheduling scheme with power level modulation is compared with the equal uplink time allocation scheme under the same network setting in Figs. 5 and 6. In Fig. 5, the distance profile of K=3 users is given as $\{5, 10, 15\}$ in meters and for $K=5, 10$, distances of $15[\text{m}]$ are added. In this setting, we confirm that regardless of the number of users, the performance of the proposed WPCN scheme is always better than that of equal resource allocation (ERA). In the proposed scheme, performance is slightly improved as the number of users increases. From (13), we have the restriction as $K \leq 11$. In the case of ERA, the performance is better when the number of users is small. This is because of the distance profile. Users from rather far distance are added, that is, the number of users with relatively poor channel power gain is increased, and UL time of users with poor channel power gain is allocated at the same ratio as users with good channel power gain. In Fig. 6, when $K = 3$, the distance is the same as the previous case, and for $K = 5, 10$, the distance of the added users is $10[\text{m}]$. In Fig. 6, the performance improvement of the proposed scheme increases, as the number of users increases, which is due to the addition of users with relatively good channel power gain compared to Fig. 5. Due to the same reason, the more the number of users in the ERA, the better the performance.

\vspace{10pt}
\section{Conclusions}
This paper proposed the power level modulation scheme for uplink time scheduling in WPCN where users are not equipped with information receiver. We formulated the sum rate maximization problems for the WPCN with the proposed power level modulation and proved the convexity of the problems. Numerical results showed that the proposed WPCN scheme significantly outperforms the conventional equal time allocation.

%








\end{document}